\documentclass[amsmath,amsfonts,aps,twocolumn,twoside,superscriptaddress]{revtex4}

\usepackage{amsfonts}
\usepackage{amsmath}
\usepackage{graphics}
\usepackage{graphicx}
\usepackage{enumerate}
\usepackage{amssymb}
\usepackage{amsthm}

\def\nfoldV{\mathbf{V}}

\newcommand{\appref}[1]{\hyperref[#1]{{Appendix~\ref*{#1}}}}
\newcommand{\be}{\begin{eqnarray} \begin{aligned}}
\newcommand{\ee}{\end{aligned} \end{eqnarray} }
\newcommand{\benn}{\begin{eqnarray*} \begin{aligned}}
\newcommand{\eenn}{\end{aligned} \end{eqnarray*}}

\newcommand*{\cC}{\mathcal{C}}
\newcommand*{\cE}{\mathcal{E}}

\newcommand*{\cI}{\mathcal{I}}

\newcommand*{\cL}{\mathcal{L}}
\newcommand*{\cM}{\mathcal{M}}
\newcommand*{\cN}{\mathcal{N}}
\newcommand*{\cD}{\mathcal{D}}
\newcommand*{\cQ}{\mathcal{Q}}

\newcommand*{\cP}{\mathcal{P}}

\newcommand{\bc}{\begin{center}}
\newcommand{\ec}{\end{center}}

\newcommand{\Tr}{\mathop{\mathrm{tr}}\nolimits}

\newtheorem{theorem}{Theorem}[section]
\newtheorem{lemma}[theorem]{Lemma}

\newcommand{\ket}[1]{|#1\rangle}
\newcommand{\bra}[1]{\langle#1|}
\def\>{\rangle}
\def\<{\langle}

\newcommand{\proj}[1]{|#1\rangle\!\langle#1|}

\newcommand{\braket}[2]{\langle #1|#2\rangle}

\begin{document}

\title{Maximal Privacy Without Coherence}
 \author{Debbie Leung}
 \affiliation{Institute for Quantum Computing, University of Waterloo, Waterloo, N2L 3G1, ON, Canada}
 \author{Ke Li}
 \affiliation{IBM TJ Watson Research Center, Yorktown Heights, NY 10598, USA}
 \affiliation{Center for Theoretic Physics,  Massachusetts Institute of Technology,
   Cambridge, MA, USA}
 \author{Graeme Smith}
 \affiliation{IBM TJ Watson Research Center, Yorktown Heights, NY 10598, USA}
 \author{John A. Smolin}
 \affiliation{IBM TJ Watson Research Center, Yorktown Heights, NY 10598, USA}

\date{\today}

\begin{abstract}
  Privacy lies at the fundament of quantum mechanics.  A
  coherently transmitted quantum state is inherently private.
  Remarkably, coherent quantum communication is not a prerequisite for
  privacy: there are quantum channels that are too noisy to transmit
  any quantum information reliably that can nevertheless send private
  classical information.  Here, we ask how much private classical
  information a channel can transmit if it has little quantum
  capacity.  We present a class of channels $\cN_d$ with input dimension
  $d^2$, quantum capacity $Q(\cN_d) \leq 1$ and private capacity
  $P(\cN_d) = \log d$.  These channels asymptotically
  saturate an interesting inequality $P(\cN) \leq \frac{1}{2}
  (\log d_A + Q(\cN))$ for any channel $\cN$ with input dimension $d_A$,
  and capture the essence of
  privacy stripped of the confounding influence of coherence.
\end{abstract}

\maketitle

Any communication link can be modeled as a (noisy) quantum channel.
Given a sender, Alice, and a receiver, Bob, a quantum channel from Alice
to Bob is a completely positive trace preserving map from an input
space $A$ to an output space $B$. The capability of a quantum channel
for communication is measured by various capacities, one for each
type of information to be transmitted.  The classical capacity $C(\cN)$
quantifies the capability of a quantum channel $\cN$ for
transmitting classical information, in bits per channel use. The private
capacity, $P(\cN)$, gives the maximum rate of private classical
communication and quantifies the optimal performance for key exchange.
Finally, the quantum capacity $Q(\cN)$, measured in
qubits per channel use, establishes the ultimate limit for
transmitting quantum information and the performance of quantum error
correction.

The three capacities mentioned above clearly satisfy $\cQ(\cN) \leq
\cP(\cN) \leq \cC(\cN)$.
The analogy between coherent transmission and privacy, and between
entanglement and a private key strongly suggest that $\cQ(\cN) =
\cP(\cN)$.  Surprisingly, it was shown in \cite{HHHO03} that not only
can $P$ and $Q$ differ, there are channels too noisy to transmit
\emph{any} quantum information (that is, $Q(\cN)= 0$) but that can
nevertheless be used to establish privacy ($P(\cN)>0$).  The central
idea of \cite{HHHO03} concerns \emph{private states} that by their
structure embody perfectly secure classical key, much as maximally
entangled pure states embody perfectly coherent correlation.

While \cite{HHHO03} draws a qualitative distinction between the
private and the quantum capacities, it remains unclear how big the
difference can be.  If the capacities were always close, then privacy
and coherence would still be closely related concepts and the
distinction would have little practical relevance.  Our contribution
is to present a class of channels with $Q(\cN_d) \leq 1$ and $P(\cN_d) =
\log d$, where $d^2$ is the input dimension.  We further
establish that such a separation is tight, by proving an inequality
\begin{align}\label{Eq:UBQ}
P(\cN) \leq \frac{1}{2}\left(\log d_A + Q(\cN)\right) \,,
\end{align}
for any channel $\cN$ with input dimension $d_A$,
quantifying the effect of incoherence in the channel transmission
on privacy: inasmuch as a channel cannot simply transmit quantum
information, it must devote half of its input space to acting as a
``shield'' system (as defined in \cite{HHHO03}).  While (\ref{Eq:UBQ})
can be inferred from properties of squashed entanglement of
quantum states \cite{Christandl-thesis,CSW10}, this particular form
appears to be new.  Our relatively simple proof involves very
different techniques.

As an aside, our channels combine features of private states from
\cite{HHHO03} and retro-correctable/random-phase-coupling
channels of \cite{BDSS06,SS09,LWZG09,SS09a} (these channels have
large assisted capacities but small $C$, $P$, and $Q$).
In addition to finding a very tight bound on $Q(\cN_d)$, we can also compute
both $P(\cN_d)$ and $C(\cN_d)$, a relative rarity in quantum information,
especially for a highly nontrivial channel.


\vspace*{2ex}
\noindent {\bf Upper Bound on Privacy}
\vspace*{1ex}

Recall that any quantum channel
can be expressed as an isometry followed by a partial trace,
$\cN(\rho) = \Tr_E U \rho U^\dagger $, where
$U \, {:} \, A \, {\rightarrow} \,BE$
with $U^\dagger U = I$.  The complementary channel acts as
$\widehat{\cN}(\rho) = \Tr_B U \rho U^\dagger$, and allows us to
define the \emph{coherent information} of a channel as
\[
Q^{(1)}(\cN) = \max_{\phi_A} I_{\rm coh}(\cN,\phi_A) \;
 {:}{=} \; \max_{\phi_A} \left[S(B)-S(E)\right],
\]
where the maximization is taken over input quantum states $\phi_A$,
and $S(B)$, $S(E)$ are the von Neumann entropies of $\rho_B = \cN(\phi_A)$
and $\rho_E = \widehat{\cN}(\phi_A)$ respectively.
In turn, the quantum capacity is proved \cite{Lloyd97,Shor02,D03} to
be the {\em regularized} coherent information: $Q(\cN) =
\lim_{n\rightarrow \infty}\frac{1}{n}Q^{(1)}(\cN^{\otimes n})$.
We say that a quantum channel $\cN$ is \emph{degradable} if
$\widehat{\cN} = \cD \circ \cN$ for some channel $\cD$
\cite{DS03} ($\cN$ can be
degraded to generate $\widehat{\cN}$).  For degradable channels,
$P(\cN) = Q(\cN) =
Q^{(1)}(\cN)$ \cite{smith2008private}.
Degradable channels also provide a powerful tool for upper bounding
the capacities of general channels \cite{AE08}.
If a
channel $\cN = \cL\circ \cM$ is a composition of two channels $\cL$
and $\cM$ with $\cM$ degradable, we have
{\small
\begin{equation}
Q(\cN) \leq P(\cN) = P\left(\cL \, {\circ} \, \cM \right)
\leq P(\cM) =  Q^{(1)}(\cM)\,.
\label{eq:longchain}
\end{equation}
}
We now have all the tools for proving Eq.~(\ref{Eq:UBQ}).
For any channel $\cN$, define $\cM$ as
\begin{align}
\cM(\rho) = \frac{1}{2}\left[\rho \otimes \proj{0} + \cN(\rho) \otimes \proj{1}\right].
\end{align}
Then, $\cN = \cL\circ \cM$ where $\cL(\sigma) =
( \cN \otimes \Pi_{0} +\cI \otimes \Pi_{1}) (\sigma)$ and
$\Pi_i(\mu) = \<i| \mu |i\>$.
To see that $\cM$ is degradable, note that the complementary channel of $\cM$ is
\begin{align}
\widehat{\cM}(\rho) = \frac{1}{2}\left[\proj{e} \otimes \proj{0} + \widehat{\cN}(\rho) \otimes \proj{1}\right],
\end{align}
where $\proj{e}$ is an orthogonal erasure flag.  Choose a degrading
map $\cD$ that first flips the flag qubit (the second register), and
then conditioned on the flag being $\ket{1}$ or $\ket{0}$, applies
$\widehat{\cN}$ to the first register or resets it to $\proj{e}$.
So, $\widehat{\cM} = \cD \circ \cM$.  Now,
applying Eq.~(\ref{eq:longchain}),
\begin{eqnarray}
 P(\cN) & \leq & Q^{(1)}(\cM)
 \nonumber \\ & = & \max_{\phi_A}\left[ S(B_1B_2)-S(E_1E_2)\right]
 \nonumber \\ & = & \max_{\phi_A}\frac{1}{2}\left[ S(\phi_A) + S\left(\cN(\phi_A)\right)-S ( \widehat{\cN}(\phi_A) )\right]
 \nonumber \\ & & \leq \frac{1}{2}\left[\log d_A + Q^{(1)}(\cN)\right].
 \nonumber \end{eqnarray}
This bound is in fact stronger than Eq.~(\ref{Eq:UBQ}), since $Q^{(1)}(\cN) \leq Q(\cN)$.


\vspace*{2ex}
\noindent {\bf Channel Construction}
\vspace*{1ex}

The family of
channels $\cN_d$ asymptotically saturating Eq.~(\ref{Eq:UBQ}) is
given by:
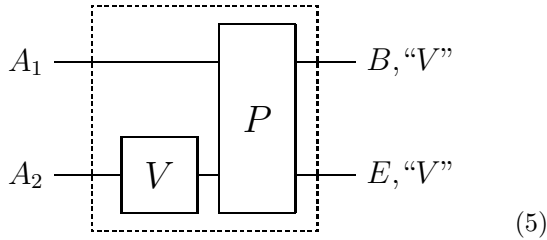
\begin{equation}
\setlength{\unitlength}{0.5mm}
\centering
\begin{picture}(130,58)
\put(30,0){\dashbox(60,60){}}
\put(05,10){\makebox(15,10){\large{$A_2$}}}
\put(05,40){\makebox(15,10){\large{$A_1$}}}
\put(20,15){\line(1,0){18}}
\put(84,45){\line(1,0){16}}
\put(20,45){\line(1,0){44}}
\put(38,05){\framebox(20,20){\Large{$V$}}}
\put(58,15){\line(1,0){6}}
\put(64,05){\framebox(20,50){\Large{$P$}}}
\put(84,15){\line(1,0){16}}
\put(100,10){\makebox(30,10){\large{$E$,``$V$''}}}
\put(100,40){\makebox(30,10){\large{$B$,``$V$''}}}
\end{picture}
\label{eq:Channel}
\end{equation}
The isometric extension of the channel $\cN_d$ is given by the
operations in the dashed box.  $\cN_d$ has two input registers $A_1$ and
$A_2$, each of dimension $d$.  A random unitary $V$ is applied to
$A_2$, followed by a controlled phase gate $P = \sum_{i,j}
\omega^{ij}\proj{i}\otimes \proj{j}$ acting on $A_1 A_2$, where
$\omega$ is a primitive $d$th root of unity.  Bob receives only $A_1$
(now relabeled $B$) and ``$V$'', which denotes a classical register
with a description of $V$.  The $A_2$ register is discarded.  The
complementary channel has outputs $A_2$ and ``$V$''.  More
formally, let $W_{V} = P \, (I\otimes V)$, $\cN_{V}(\rho) = \Tr_E W_V\rho
W_V^\dagger$, and $\cN_d = \mathbb{E}_V \cN_{V} \otimes \proj{V}_{V_B}$,
where the register $V_B$ holds ``$V$''.  The isometric extension is
given by
\[U_d \ket{\psi}_{A_1A_2} = \sum_V \! \sqrt{{\rm pr}(V)} \,
(W_{V}\ket{\psi}_{A_1A_2}) \otimes \ket{V}_{V_B}\otimes
\ket{V}_{V_E}\]
and the complementary channel acts as $\hat{\cN}_d(\rho)
= \Tr_{BV_B}U_d \rho U_d^\dagger = \mathbb{E}_V \Tr_B W_V \rho W_V^\dagger
\otimes \proj{V}_{V_B}$.

Here is the intuition behind the construction:  The classical capacity
of this channel is at least $\log d$, since the $d$ computation
basis states of $A_1$ are transmitted
without error.  Furthermore, we will see that
inserting a maximally mixed state into $A_2$ keeps the environment
ignorant of the transmitted message so $P(\cN_d)\geq \log
d$.  However, as the classical capacity is no greater than the output
entropy, and ``$V$'' is independent of the input, $C(\cN_d)\leq \log d$,
so, $C(\cN_d) = P(\cN_d) = \log d$.  However, the channel is
constructed to suppress the quantum capacity, since without knowing
$V$, Alice cannot avoid the $P$ gate from entangling $A_1$ with $A_2$,
thereby dephasing $A_1$.  We will prove $Q(\cN_d)\leq 1$.

Our proofs of the above statements hold for any
$V$ drawn from a so-called exact unitary $2$-design, and thus,
$V$ can be a random Clifford gate \cite{DLT01}.
In our work to {\em lower} bound $Q(\cN)$,
a Haar distributed $V$ is analyzed as a first attempt.  We expect
a similar conclusion for random Clifford gate $V$.

\vspace*{2ex}
\noindent {\bf Private Capacity}
\vspace*{1ex}

For an ensemble $\cE = \{p_i,\phi_i\}$ and
channel $\cN$, the private information is defined as
\begin{align}
P^{(1)}(\cN,\cE) = \chi(\cN,\cE)-\chi(\widehat{\cN},\cE),
\end{align}
with Holevo information $\chi(\cN,\cE) = S(\rho) - \sum_i
p_iS(\rho_i)$ evaluated on the induced ensemble $\cN(\cE) =
\{p_i,\rho_i = \cN(\phi_i)\}$ and average state $\rho = \sum_i p_i
\rho_i$ (similarly for $\chi(\widehat{\cN},\cE)$). For any ensemble
$\cE$, $P^{(1)}(\cN,\cE)$ is an achievable rate for private
communication and thus a lower bound on $P(\cN)$ \cite{D03}.

For our channel $\cN_d$, choosing probabilities $p_i = \frac{1}{d}$ and
states $\phi_i = \proj{i}_{A_1}\otimes (\frac{I}{d})_{A_2}$ for $i = 1\dots d$,
gives $\chi(\cN_d,\cE) = \log d$ and $\chi(\widehat{\cN}_d,\cE) =
0$, so $P(\cN_d) \geq \log d$, as claimed.  Readers familiar with
\cite{HHHO03} will notice that, the Choi-state of $\cN_d$ with Alice
holding $R_1,R_2$ is an exact private state of key system $R_1 B$ and
shield $R_2$.


\vspace*{2ex}
\noindent {\bf Upper Bound on Quantum Capacity}
\vspace*{1ex}

To get an upper bound on $Q(\cN_d)$, we consider the asymptotic
behavior of the coherent information, $Q^{(1)}(\cN_d^{\otimes n})$ for
arbitrarily large $n$.  We first define suitable notations.  We group
together the first input $A_1$ from all $n$ channel uses, call it
$A_1^n$, and we similarly define $A_2^n$, $B^n$, $V_B^n$, and $V_E^n$.
We use $\mathbf{x}$ to denote an $n$-tuple of integers $(x_1,x_2,\cdots,
x_n)$ where each $x_i$ has range $\{0,1,\cdots,d{-}1\}$, and similarly for
 $\mathbf{y}$.
Finally, a random $V$ is drawn from each channel use, and we
denote the tensor product of $n$ such independent and identically
drawn unitaries by $\nfoldV$.

We consider the $n$-shot coherent information
$Q^{(1)}(\cN_d^{\otimes n}) = \max_{\phi_{A_1^n A_2^n}}
  \left[S(B^n V_B^n)-S(E^n V_E^n)\right]$.  Since Bob and the environment
receives the same classical description ``$V$'',
$Q^{(1)}(\cN_d^{\otimes n}) = \max_{\phi_{A_1^n A_2^n}}
  \left[S(B^n|V_B^n)-S(E^n|V_E^n)\right]$.
First, we show that the optimal input state has a special form.


\begin{lemma}
\label{Lemma:StdForm}
For the channel $\cN_d$ of Eq~(\ref{eq:Channel}), the coherent information
$I_{\rm coh}(\cN_d^{\otimes n},\phi_{A_1^nA_2^n})$ is maximized on states of the form
\begin{align}
\phi_{A_1^nA_2^n} = \sum_{\mathbf{x}} p_{\mathbf{x}}\proj{\mathbf{x}}_{A_1^n} \otimes \proj{\varphi^{\mathbf{x}}}_{A_2^n},
\label{eq:optin}
\end{align}
where $\mathbf{x} = (x_1,\cdots, x_n)$ and
$\ket{\mathbf{x}} = \otimes_{i=1}^n\ket{x_i}$
is a standard basis state on $A_1^n$.
\end{lemma}

\begin{proof}
First, we show that the optimal state has the form
\begin{align}\label{Eq:CQmixed}
\sigma_{A^n_1A^n_2} = \sum_{\mathbf{x}} p_{\mathbf{x}}\proj{\mathbf{x}}_{A_1^n} \otimes \varphi^{\mathbf{x}}_{A_2^n},
\end{align}
where $\varphi^{\mathbf{x}}_{A_2^n}$ is potentially mixed.
To see this,
let $\psi_{A_1^nA_2^n}$ be an arbitrary input state,
and
$\sigma_{A_1^nA_2^n} = \left(\cP^{\otimes n}\otimes
I_{A_2^n}\right)(\psi_{A_1^nA_2^n})$ where
$\cP(\rho)
= \frac{1}{d}\sum_{i=0}^{d-1}Z_i \rho Z_i^\dagger$
is the completely dephasing map.  So, $\sigma_{A_1^nA_2^n}$
indeed has the form given by Eq.~(\ref{Eq:CQmixed}).
Now,
\[
\begin{split}
     I_{\rm coh}(\cN_d^{\otimes n}, \sigma_{A_1^nA_2^n})
  =  I_{\rm coh}(\cN_d^{\otimes n} \circ (\cP^{\otimes n} \! \otimes \! I_{A_2^n}), \psi_{A_1^nA_2^n})
\nonumber \\
  =  I_{\rm coh}(\cP^{\otimes n} \circ \cN_d^{\otimes n}, \psi_{A_1^nA_2^n})
\geq I_{\rm coh}(\cN_d^{\otimes n}, \psi_{A_1^nA_2^n})
\end{split}
\]
since $\cP$ commutes with $\cN_d$; and
$\cP$ is unital, thus the entropy cannot decrease.
Meanwhile the reduced state on $E^nV_E^n$ remains the same,
so the coherent information cannot decrease.

Next, we show that $\varphi^{\mathbf{x}}_{A_2^n}$ in
Eq.~(\ref{Eq:CQmixed}) can be taken to be pure.
Fix an arbitrary $\mathbf{x}$.
Let $\varphi^{\mathbf{x}}_{A_2^n} {=} \sum_w q(w|\mathbf{x})
\proj{\mu_w^{\mathbf{x}}}$, and for each $w$, let
\begin{align}
\eta_{A_1^nA_2^n}^{\mathbf{x},w} = p_{\mathbf{x}}\proj{\mathbf{x}}\otimes \proj{\mu_w^{\mathbf{x}}}
+ \sum_{\mathbf{y}\neq \mathbf{x}} p_{\mathbf{y}}\proj{\mathbf{y}}\otimes \varphi^{\mathbf{y}}_{A_2^n} \,.
\nonumber
\end{align}
We now show that
\begin{equation} \exists w'  ~{\rm s.t.}~
I_{\rm coh}(\cN_d^{\otimes n},\eta_{A_1^nA_2^n}^{\mathbf{x},w'})
\geq I_{\rm coh}(\cN_d^{\otimes n},\sigma_{A_1^nA_2^n}) \,.
\label{eq:purestateopt}
\end{equation}
To see this, note that
for each $\mathbf{x}$ and $w$,
$\cN^{\otimes n}_d\left( \sigma_{A_1^nA_2^n} \right)
= \cN^{\otimes n}_d \left( \eta_{A_1^nA_2^n}^{\mathbf{x},w}\right)$,
so those states have the same entropy.
For the complementary channel, observe that by construction,
\[
\sigma_{A_1^nA_2^n} = \sum_{w} q(w|\mathbf{x}) \, \eta_{A_1^nA_2^n}^{\mathbf{x},w}
\]
so $\hat{\cN}_d^{\otimes n}(\sigma_{A_1^nA_2^n})
= \sum_w q(w|\mathbf{x}) \; \hat{\cN}^{\otimes n}_d \! \left( \eta_{A_1^nA_2^n}^{\mathbf{x},w}\right)$,
and by concavity of entropy,
\[
\begin{split}
& S(E^n V^n_E)_{\hat{\cN}_d^{\otimes n}\left(\sigma_{A_1^nA_2^n}\right)} \\
& \geq \sum_w q(w|\mathbf{x})S(E^n V^n_E)_{\hat{\cN}^{\otimes n}_d
      \left( \eta_{A_1^nA_2^n}^{\mathbf{x},w}\right)}
\end{split}
\]
so Eq.~(\ref{eq:purestateopt}) holds.  Iterating this process
gives an optimal state of the form given by Eq.~(\ref{eq:optin}).
\end{proof}

The optimal input in Eq.~(\ref{eq:optin}) gives a conditional
output entropy of $S(B^n|V_B^n) = H(p_\mathbf{x}) = -
\sum_{\mathbf{x}}p_\mathbf{x}\log p_\mathbf{x}$.
We need to compare
this to the conditional output entropy of the environment,
$S(E^n|V_E^n)$.
To do so,
note that the controlled-phase gate for one channel can be expressed
as $P = \sum_i \proj{i} \otimes Z_i$ where $Z_i \ket{j} = \omega^{ij}
\ket{j}$.  For the $n$-tuple $\mathbf{x}$,
we define $Z^{\bf x} = Z_{x_1}\otimes ... \otimes Z_{x_n}$.
The output state on the environment is
\begin{align}
\rho_{E^n,V_E^n} = \mathbb{E}_{\nfoldV}\sum_{\mathbf{x}} \,
 p_{\mathbf{x}} \,
 \rho_{E^n}^{\mathbf{x},\nfoldV} \otimes \proj{\nfoldV}_{V_E^n},
\end{align}
where
\be
\rho_{E^n}^{\mathbf{x},\nfoldV} =
Z^{\mathbf{x}}\nfoldV\proj{\varphi^{\mathbf{x}}} \nfoldV^\dagger
Z^{-\mathbf{x}} \,.
\label{eq:outputcondv}
\ee
If, for fixed $\nfoldV$ and
$\mathbf{x}\neq \mathbf{y}$, $\rho_{E^n}^{\mathbf{x},\nfoldV}$ and
$\rho_{E^n}^{\mathbf{y},\nfoldV}$ were orthogonal, we would have
$S(E^n|V_E^n) = H(p_{\mathbf{x}})$ and
$Q^{(1)}(\cN_d^{\otimes n}) = 0$.  Instead, we
show that the Hilbert-Schmidt inner product
$\mathbb{E}_{\nfoldV} \Tr
\rho_{E^n}^{\mathbf{x},\nfoldV}\rho_{E^n}^{\mathbf{y},\nfoldV}$ is
low on average (over choice of $\nfoldV$):

\begin{lemma} \label{lem:littleoverlap}
The states $\rho^{\mathbf{x},\nfoldV}_{E^n}$ given by
Eq.~(\ref{eq:outputcondv}) satisfies
\[\mathbb{E}_{\nfoldV}\Tr \rho^{\mathbf{x},\nfoldV}_{E^n} \rho^{\mathbf{y},\nfoldV}_{E^n}
 \leq \frac{1}{(d-1)^{d_H(\mathbf{x},\mathbf{y})}} \,, \]
where $d_H(\mathbf{x},\mathbf{y}) = |\{i| x_i \neq y_i\}|$ is the
Hamming distance between $\mathbf{x}$ and $\mathbf{y}$.
\end{lemma}

Next we derive a lower bound on the output entropy of the environment,
by considering the R\'enyi-2 entropy of many copies of
$\rho_{E^n}^{\nfoldV}$ using Lemma \ref{lem:littleoverlap}.

\begin{lemma}\label{Lemma:Renyi2}
For an input given by Eq.~(\ref{eq:optin}),
conditioned on $\nfoldV$,
the output state on the environment
$\rho_{E^n}^{\nfoldV}  = \sum_{\mathbf{x}\in[d]^n} p_{\mathbf{x}} \rho^{\mathbf{x},\nfoldV}_{E^n}$
satisfies
\begin{align}
\mathbb{E}_{\nfoldV}\Tr \left(\rho_{E^n}^{\nfoldV}\right)^2 \leq 2^n \sum_{\mathbf{x}} p_{\mathbf{x}}^2.
\label{eq:statement1}
\end{align}
and
\begin{align}
S(E^n|V_E^n)= \mathbb{E}_{\nfoldV} S \left( \rho_{E^n}^{\nfoldV}\right) \geq H(\mathbf{X}) - n.
\label{eq:statement2}
\end{align}
\end{lemma}

\begin{proof} To prove the first statement Eq.~(\ref{eq:statement1}),
\begin{align}
\mathbb{E}_{\nfoldV}\Tr \left(\rho_{E^n}^{\nfoldV}\right)^2 &=
\mathbb{E}_{\nfoldV} \sum_{\mathbf{x},\mathbf{y}}p_{\mathbf{x}}p_{\mathbf{y}}
                    \Tr \left( \rho_{E^n}^{\mathbf{x},\nfoldV} \rho_{E^n}^{\mathbf{y},\nfoldV} \right)
\nonumber
\\
& \leq \sum_{\mathbf{x},\mathbf{y}}p_{\mathbf{x}}p_{\mathbf{y}} \; \frac{1}{(d{-}1)^{d_H(\mathbf{x},\mathbf{y})}}
\nonumber
\\
& = \sum_{w=0}^{n} \frac{1}{(d{-}1)^w}\sum_{\mathbf{x}}\sum_{\mathbf{y}| d_H(\mathbf{x},\mathbf{y})= w}p_{\mathbf{x}}p_{\mathbf{y}}
\nonumber
\\
& = \sum_{w=0}^{n} \frac{1}{(d{-}1)^w}\sum_{|\mathbf{e}| = w} \sum_{\mathbf{x}}  p_{\mathbf{x}}p_{\mathbf{x} + \mathbf{e}} \,,
\nonumber
\end{align}
where the first inequality follows from Lemma \ref{lem:littleoverlap}.
By the Cauchy-Schwartz inequality,
\[\sum_{\mathbf{x}}  p_{\mathbf{x}}p_{\mathbf{x} {+} \mathbf{e}}
\leq \sqrt{\sum_{\mathbf{x}}  p_{\mathbf{x}}^2} \sqrt{\sum_{\mathbf{x}}  p_{\mathbf{x}{+}\mathbf{e}}^2}\\
= \sum_{\mathbf{x}}  p_{\mathbf{x}}^2 .\]
We thus have
\begin{align}
& \mathbb{E}_{\nfoldV}\Tr \left(\rho_{E^n}^{\nfoldV}\right)^2 \leq
\sum_{w=0}^{n} \frac{1}{(d-1)^w}\sum_{|\mathbf{e}| = w} \sum_{\mathbf{x}}  p_{\mathbf{x}}^2
\nonumber
\\
\nonumber
& = \sum_{w=0}^{n} \frac{1}{(d-1)^w} \binom{n}{w} (d-1)^w \sum_{\mathbf{x}}  p_{\mathbf{x}}^2
= 2^n \sum_{\mathbf{x}}  p_{\mathbf{x}}^2
\,.
\end{align}
To prove the second statement Eq.~(\ref{eq:statement2}),
consider the case that we input $m$ copies of the state given
by Eq.~(\ref{eq:optin}) into $mn$ copies of the channel $\cN_d$.
Conditioning on the random unitaries of the channels, this
results in an output state $\rho_{E^{nm}}^{\nfoldV^m}=
\otimes_{i=1}^m \rho_{E^n}^{\nfoldV_i}$ on the environment.
The state $\rho_{E^{nm}}^{\nfoldV^m}$ involves a mixture over
$(\mathbf{X})^m$, which is $m$ i.i.d.\ copies of $\mathbf{X}$.
Let $\tilde{\mathbf{X}}_m$ be the restriction of $(\mathbf{X})^m$
to its typical set $\{\mathbf{x}^m:|-\frac{1}{m}\log
p_{\mathbf{x}^m}-S(\mathbf{X})| \leq \epsilon\}$.  Then the state
$\rho_{E^{nm}}^{\nfoldV^m}$ can be well approximated by
$\tilde{\rho}_{E^{nm}}^{\nfoldV^m}$, where the mixture is only
taken over $\tilde{\mathbf{X}}_m$. Specifically, using the properties
of typical sets (cf.~\cite{CT91}, for example), we can easily check
that for $m$ big enough $\|\rho_{E^{nm}}^{\nfoldV^m}-
\tilde{\rho}_{E^{nm}}^{\nfoldV^m}\|_1 \leq 2\epsilon$. Thus we have
\begin{align}
m \mathbb{E}_\nfoldV S \left(\rho_{E^n}^{\nfoldV} \right)
& = \mathbb{E}_{\nfoldV^m} S \left(\rho_{E^{nm}}^{\nfoldV^{m}}  \right) \label{eq:doubleblock} \\
& \geq \mathbb{E}_{\nfoldV^m} S \left(\tilde{\rho}_{E^{nm}}^{\nfoldV^{m}}\right) -
                                  \epsilon mn \log d-H((\epsilon, 1-\epsilon)), \nonumber
\end{align}
where for the inequality we have used the strengthened Fannes inequality~\cite{Aud07}.
Furthermore, Eq.~(\ref{eq:statement1}) applies to
$\tilde{\rho}_{E^{nm}}^{\nfoldV^{m}}$, so that
\begin{align}
\mathbb{E}_{\nfoldV^m} \Tr \left( \tilde{\rho}_{E^{nm}}^{\nfoldV^{m}} \right)^2
\leq 2^{mn} \sum_{\mathbf{x}^m \in \tilde{\mathbf{X}}_m } \Big(\frac{p_{\mathbf{x}^m}}
{\sum_{\mathbf{x}^m \in \tilde{\mathbf{X}}_m } p_{\mathbf{x}^m}}\Big)^2.
\label{eq:applystatement1}
\end{align}
By convexity of ${-}\log$, Eq.~(\ref{eq:applystatement1}) translates to
\begin{align}
\mathbb{E}_{\nfoldV^m}S_2 \left( \tilde{\rho}_{E^{nm}}^{\nfoldV^{m}} \right) \geq
S_2 (\tilde{\mathbf{X}}_m) - mn \, ,
\label{eq:renyiversion}
\end{align}
where the R\'enyi-2 entropy $S_2$ is defined as $S_2(\rho):=-\log \Tr \rho^2$.
Since $\tilde{\mathbf{X}}_m$ only has typical sequences, $S_2 (
\tilde{\mathbf{X}}_m )\, {\approx} \, m S(\mathbf{X})$. Here we only
need an inequality, specifically, for $m$ sufficiently large
\begin{align}
S_2(\tilde{\mathbf{X}}_m)\geq m S(\mathbf{X})-3m\epsilon+2\log(1-\epsilon),
\label{eq:renyi-2v1}
\end{align}
which can be easily confirmed using the properties of typical sets.
Finally, connecting Eqs.~(\ref{eq:doubleblock}) and
(\ref{eq:renyiversion}) by the well-known relation between R\'enyi
entropies, $S \geq S_2$, and also using Eq.~(\ref{eq:renyi-2v1}),
we have
\[\begin{split} \mathbb{E}_{\nfoldV}S \left(\rho_{E^n}^{\nfoldV} \right)
\geq S(\mathbf{X}) &- n - \epsilon \, n \log d - 3\epsilon  \\
                     &-\frac{1}{m}H((\epsilon,1-\epsilon))+\frac{2}{m}\log(1-\epsilon)\,. \end{split}\]
Taking $\epsilon \rightarrow 0$ gives the desired result.
\end{proof}

\noindent Together, $Q^{(1)}(\cN_d^{\otimes n}) \leq n$, so,
$Q(\cN_d) \leq 1$.

\vspace*{2ex}

When proving the upper bound on $Q$, we cannot assume apriori the
entropy of $B^n V_B^n$ is maximal for the optimal input, ruling out
the simpler path to show that the entropy of $E^n V_E^n$ is maximal.
Instead we have to show that $S(B^n V_B^n)-S(E^n V_E^n)$ is small for
all distributions.  Perhaps our technique has other applications.
Also Lemma \ref{Lemma:Renyi2} effectively converts a statement
concerning the R\'enyi-2 entropy into an analogue for the entropy for a
large family of states, which may be of interest elsewhere.

\vspace*{2ex}
\noindent {\bf Achievable Quantum Rate}
\vspace*{1ex}

We have shown that $Q(\cN_d) \leq 1$, but could it actually be $0$?
It turns out not: by choosing a specific input state and evaluating
the associated coherent information, we can obtain the explicit lower
bound $Q(\cN_d) \geq (1{-}\gamma)\log_2 e \approx 0.61$ as $d
\,{\rightarrow} \, \infty$, where $\gamma$ $=$ $\lim_{t \rightarrow \infty}
\left( \sum_{k=1}^t {1 \over k} - \ln t \right)$ is the
Euler-Mascheroni constant.  In the appendix, we derive this
lower bound of the coherent information by considering the input
$\phi_{A_1A_2} = (\frac{I}{d})_{A_1} \otimes \proj{0}_{A_2}$.

\vspace*{2ex}
\noindent {\bf Discussion}
\vspace*{1ex}

In \cite{HHHO03} it was shown that privacy and
distillable entanglement can be different, indeed privacy can be
nonzero even for bound-entangled states.  What we have shown is
similar, but somewhat incomparable.  Our result is stronger in
that the separation is maximal, saturating Eq.~(\ref{Eq:UBQ}), but
it only applies to the channel case, implicitly not allowing
classical communication.  The two-way assisted quantum capacity
$Q_2(\cN_d)$ is maximal (not zero!) and equal to the private capacity $\log d$.
An open question is how big the separation can be in the two-way setting?
\\[2ex] \noindent {\bf Funding}
DL was supported by NSERC, DAS, CRC, and CIFAR, KL by NSF,
GS and JS by DARPA QUEST.





\appendix

\section{Low quantum capacity}

\begin{lemma}
\label{lem:littleoverlapappendix}
Define, on $A_2^n$, the state
\begin{align}
\label{eq:rhoExV}
\rho^{\mathbf{x},\nfoldV}_{E^n} = Z^{\mathbf{x}}\nfoldV\proj{\varphi^{\mathbf{x}}} \nfoldV^\dagger Z^{-\mathbf{x}}.
\end{align}
Then,
\begin{align}
\mathbb{E}_{\nfoldV}\Tr \rho^{\mathbf{x},\nfoldV}_{E^n} \rho^{\mathbf{y},\nfoldV}_{E^n}\leq \frac{1}{(d-1)^{d_H(\mathbf{x},\mathbf{y})}},
\end{align}
where $d_H(\mathbf{x},\mathbf{y}) = |\{i| x_i \neq y_i\}|$ is the Hamming distance between $\mathbf{x}$ and $\mathbf{y}$.
\end{lemma}
\begin{proof}

First note that
\begin{align}
& \mathbb{E}_{\nfoldV}\Tr\left[ \rho^{\mathbf{x},\nfoldV}_{E^n}\rho^{\mathbf{y},\nfoldV}_{E^n}\right]
\nonumber
\\
& =
\mathbb{E}_{\nfoldV} \Tr\left[
Z^{\mathbf{x}}\nfoldV\proj{\varphi^{\mathbf{x}}}\nfoldV^\dagger Z^{-\mathbf{x}}
Z^{\mathbf{y}}\nfoldV\proj{\varphi^{\mathbf{y}}} \nfoldV^\dagger Z^{-\mathbf{y}}
\right]
\nonumber
\\
& =
\mathbb{E}_{\nfoldV} ~
\<\varphi^{\mathbf{y}}|  \nfoldV^\dagger Z^{-\mathbf{y}} Z^{\mathbf{x}}\nfoldV |\varphi^{\mathbf{x}}\>
\<\varphi^{\mathbf{x}}|  \nfoldV^\dagger Z^{-\mathbf{x}} Z^{\mathbf{y}}\nfoldV  |\varphi^{\mathbf{y}}\>
\nonumber
\\
& =
\mathbb{E}_{\nfoldV} ~
\<\varphi^{\mathbf{y}}| \<\varphi^{\mathbf{x}}| \left( \nfoldV^\dagger Z^{\mathbf{x}-\mathbf{y}}\nfoldV \right)
\otimes \left( \nfoldV^\dagger Z^{\mathbf{y}-\mathbf{x}} \nfoldV \right) |\varphi^{\mathbf{x}}\>  |\varphi^{\mathbf{y}}\>
\nonumber
\\
& = \bra{\varphi^{\mathbf{y}}}\bra{\varphi^{\mathbf{x}}} S^{\mathbf{x}-\mathbf{y}} \ket{\varphi^{\mathbf{x}}}\ket{\varphi^{\mathbf{y}}},
\label{Eq:Sprod}
\end{align}
where
\begin{align}
S^{\mathbf{x}-\mathbf{y}} = \mathbb{E}_{\nfoldV}
\left[
\nfoldV^\dagger \! \otimes \! \nfoldV^\dagger
\left( Z^{\mathbf{x}-\mathbf{y}} \otimes Z^{\mathbf{y}-\mathbf{x}} \right)
\nfoldV \! \otimes \! \nfoldV
\right] \,.
\end{align}
We now evaluate $S^{\mathbf{x}-\mathbf{y}}$ which equals $\otimes_{i = 1}^n S^{x_i-y_i}$
for
\begin{align}
S^{x_i-y_i} = \mathbb{E}_{V} V^\dagger \otimes  V^\dagger \left( Z^{x_i-y_i} \otimes Z^{-(x_i-y_i)} \right) V \otimes V \,.
\nonumber
\end{align}
If $x_i-y_i = 0$, then $S^{0} = \mathbb{I}$.
If $x_i-y_i = a \neq 0$, then,
\begin{align}
& S^{a} = \mathbb{E}_{V} V^\dagger \otimes V^\dagger (Z^{a} \otimes Z^{-a}) V \otimes V \nonumber \\
& = \Tr \left[\Pi_\mathrm{sym} Z^a \! \otimes \! Z^{-a}\right] \frac{\Pi_\mathrm{sym}}{d_\mathrm{sym}}
  + \Tr \left[\Pi_\mathrm{anti} Z^a \! \otimes \! Z^{-a}\right] \frac{\Pi_\mathrm{anti}}{d_\mathrm{anti}} \nonumber
\end{align}
where $\Pi_\mathrm{sym}$ and $\Pi_\mathrm{anti}$ are
the projectors onto the symmetric and the antisymmetric subspaces.
More specifically, let $F = \sum_{i,j} = \ket{i}\ket{j}\bra{j}\bra{i}$ be the swap operator.  Then,
\begin{align}
\Pi_\mathrm{sym} & = \frac{1}{2}\left(\mathbb{I} + F \right)\\
\Pi_\mathrm{anti}& = \frac{1}{2}\left(\mathbb{I} - F \right).
\end{align}
Note that $\Tr [Z^a \! \otimes \! Z^{-a}] = 0$, and
\begin{align}
& \Tr \left[F (Z^a \! \otimes \! Z^{-a}) \right]  \nonumber \\ \nonumber
& ~~~~~ = \sum_{i,j}\omega^{ai- aj} \Tr\left[ \ket{i}\bra{j} \otimes \ket{j}\bra{i} \right] =  \sum_{i} 1 = d.
\end{align}
So, $\Tr \left[\Pi_\mathrm{sym} Z^a \! \otimes \! Z^{-a}\right] \, {=} \, \frac{d}{2}$,
$\Tr \left[\Pi_\mathrm{anti} Z^a \! \otimes \! Z^{-a}\right] \, {=} \, {-}\frac{d}{2}$.
Consequently,
\begin{align}
S^a  = \frac{d}{2}\left[ \frac{\Pi_\mathrm{sym}}{d_\mathrm{sym}} - \frac{\Pi_\mathrm{anti}}{d_\mathrm{anti}}\right]
	= \frac{\Pi_\mathrm{sym}}{d+1} - \frac{\Pi_\mathrm{anti}}{d-1} \,.
\label{eq:sa}
\end{align}

With Eqs.~(\ref{Eq:Sprod}) and (\ref{eq:sa}), we can finish the proof in two steps.
From Eq.~(\ref{eq:sa}),  $\| S^{a} \|_{\infty} = \frac{1}{d-1}$ and
$\| S^{\mathbf{x}-\mathbf{y}} \|_\infty = (d{-}1)^{{-}d_H(\mathbf{x},\mathbf{y})}$.
Furthermore, for any hermitian operator $H$ and unit vectors $|\psi\>$ and $|\phi\>$,
$|\bra{\psi} H \ket{\phi}| \leq \| H \|_\infty$ (see Lemma \ref{lem:norms}).
Applying these to Eq.~(\ref{Eq:Sprod}) gives
\[
\mathbb{E}_{\nfoldV}\Tr\left[ \rho^{\mathbf{x},\nfoldV}_{E^n}\rho^{\mathbf{y},\nfoldV}_{E^n}\right]
 \leq  \| S^{\mathbf{x}-\mathbf{y}} \|_{\infty} = \frac{1}{(d-1)^{d_H(\mathbf{x},\mathbf{y})}}.
\]
\end{proof}

\begin{lemma} \label{lem:norms}
For any hermitian operator $H$ and unit vectors $|\psi\>$ and $|\phi\>$,
$|\bra{\psi} H \ket{\phi}| \leq \| H \|_\infty$.
\end{lemma}

\begin{proof}
This is an elementary result in matrix analysis, but we provide a
proof for completeness.  Let $H = \sum_{i}\lambda_i |\mu_i\>\<\mu_i|$
be its spectral decomposition.  Then,
\begin{align}
\| H \ket{\phi} \|_2^2 & = \bra{\phi} H^\dagger H \ket{\phi} \nonumber \\
 & = \sum_{i}\lambda_i^2 \braket{\phi}{\mu_i} \braket{\mu_i}{\phi} \nonumber  \\
& \leq \lambda_{\rm max}^2 \sum_{i} \braket{\phi}{\mu_i} \braket{\mu_i}{\phi}
=  \| H \|_\infty^2 \,. \nonumber
\end{align}
Furthermore, by the Cauchy-Schwartz inequality,
\begin{align}
|\bra{\psi} H \ket{\phi}| \leq \| \bra{\psi}\|_2 \| H \ket{\phi}\|_2 \leq  \| H \|_\infty.
\end{align}
\end{proof}

\section{Average entropy of dephased random states}

The following lemma lets us evaluate the entropy of the environment when the input
state is $\phi_{A_1A_2} = (\frac{I}{d})_{A_1} \otimes\proj{0}_{A_2}$.

\begin{lemma}
  \label{lemma:average-entropy}
  Let $\ket{\varphi}$ be a normalized quantum state drawn randomly from
  the Hilbert space $\mathbb{C}^d$, according to the unitarily-invariant
  probability measure, i.e., $\ket{\varphi}=U\ket{0}$ with $U$ a Haar-random
  unitary operator. Let $\mathcal{P}$ be the completely dephasing quantum
  operation, namely, $\mathcal{P}(\rho):=\sum_{k=0}^{d-1}\bra{k}\rho\ket{k}\proj{k}$.
  We have
  \[
    \mathbb{E}_\varphi S(\mathcal{P}(\proj{\varphi}))=(\log e) (H_d-1),
  \]
  where $H_d:=\sum_{k=1}^d\frac{1}{k}$ is the $d\text{th}$ harmonic number.
\end{lemma}

\begin{proof}
Write $\ket{\varphi}=\sum_{k=0}^{d-1}(x_k+iy_k)\ket{k}=\ket{\varphi(\vec{v})}$,
with $\vec{v}:=(x_0, x_1,\ldots x_{d-1},y_0,  y_1,\ldots y_{d-1})^\mathrm{T}
\in \mathbb{R}^{2d}$ being a vector on the $(2d-1)$-dimensional unit sphere
$\mathsf{S}^{(2d-1)}$. For any unitary operator $V=A+iB$ acting on $\mathbb{C}^d$,
with $A$ and $B$ being the real part and imaginary part, respectively, it is
straightforward to check that
\[
  V\ket{\varphi(\vec{v})}=\ket{\varphi(O\vec{v})},
\]
where
\[
O=\begin{pmatrix} A & -B \\ B & A \end{pmatrix}
\]
is an orthogonal operator acting on $\mathbb{R}^{2d}$. Thus, letting
$\mu$ be the unitarily-invariant probability measure on the set of normalized pure
states, we find that, on the parameter set $\{\vec{v}|\vec{v}\in \mathsf{S}^{(2d-1)}\}$,
it translates to the orthogonally-invariant measure, which in turn is proportional to
the Euclidean volume of the corresponding portion on $\mathsf{S}^{(2d-1)}$. Denote the
volume of $\mathsf{S}^{(l)}$ as $S^{(l)}$, and the volume elements on $\mathsf{S}^{(l)}$ as
$\mathrm{d}S^{(l)}$. As a result of the above argument we have
\begin{equation}
  \label{eq:measure-volume-element}
  \mu(\mathrm{d}\varphi)\equiv \mu(\mathrm{d}\vec{v})=\frac{1}{S^{(2d-1)}}\mathrm{d}S^{(2d-1)}.
\end{equation}
Note $S(\mathcal{P}(\proj{\varphi}))\; {=}\; -\sum_{k=0}^{d-1}|\braket{k}{\varphi}|^2
\log |\braket{k}{\varphi}|^2$. By symmetry and Eq.~(\ref{eq:measure-volume-element}) we have
\begin{equation}\begin{split}
  \label{eq:average-entropy-1}
    &\mathbb{E}_\varphi S(\mathcal{P}(\proj{\varphi})) \\
    &=-d \; \frac{1}{S^{(2d-1)}}\int(x_0^2+y_0^2)\log(x_0^2+y_0^2)\,\mathrm{d}S^{(2d-1)} \,.
\end{split}\end{equation}
Let $\cos\theta=\sqrt{x_0^2+y_0^2}$, with $\theta\in[0, \frac{\pi}{2}]$.
Then $\sin\theta=\sqrt{\sum_{k=1}^{d-1}(x_k^2+y_k^2)}$. Fixing $\theta$, the changing of
$(x_0, y_0)$ forms a circle (one-dimensional sphere) of radius $\cos\theta$, and we
denote it as $\mathsf{S}_1$; the changing of $(x_1,x_2,\ldots x_{d-1}, y_1,y_2,\ldots y_{d-1})$
forms a $(2d-3)$-dimensional sphere $\mathsf{S}_2$ of radius $\sin\theta$. The volume
elements of these two spheres are thus, respectively,
\[\mathrm{d}S_1=\cos\theta\,\mathrm{d}S^{(1)}, \qquad
  \mathrm{d}S_2=\sin^{2d-3}\theta\,\mathrm{d}S^{(2d-3)}.
\]
On the other hand, fixing the other spherical coordinates of $\mathsf{S}_1$ and
$\mathsf{S}_2$, the changing of $\theta$ forms a quarter of a unit circle,
$\mathsf{S}_3:=\{(\cos\theta, \sin\theta)|0\leq\theta\leq\frac{\pi}{2}\}$. The volume
element of $\mathsf{S}_3$ is obviously
\[\mathrm{d}S_3=\mathrm{d}\theta.\]
We further observe that $\mathrm{d}S_1$, $\mathrm{d}S_2$ and $\mathrm{d}S_3$ are mutually
orthogonal on the big sphere $\mathsf{S}^{(2d-1)}$, because $\mathrm{d}S_1$ and
$\mathrm{d}S_2$ fall in two distinct orthogonal subspaces, respectively, and
$\mathrm{d}S_3$ falls in the radial directions of both $\mathsf{S}_1$ and
$\mathsf{S}_2$. Hence we have
\begin{eqnarray}
  \label{eq:volume-element}
  \mathrm{d}S^{(2d-1)}&=&\mathrm{d}S_1\,\mathrm{d}S_2\,\mathrm{d}S_3
\\
    &=&\cos\theta \, \sin^{2d-3}\theta \,\mathrm{d}\theta \,\mathrm{d}S^{(1)}\,\mathrm{d}S^{(2d-3)} \,.
\nonumber
\end{eqnarray}
Using Eq.~(\ref{eq:volume-element}),
and substituting $\cos^2\theta$ for $x_0^2+y_0^2$, we find that the right
hand side of
Eq.~(\ref{eq:average-entropy-1}) equals
\begin{equation}
  \label{eq:average-entropy-2}
  - d \frac{S^{(1)}S^{(2d-3)}}{S^{(2d-1)}} \int_{\theta=0}^{\frac{\pi}{2}}
     \!\!  \cos^3\!\theta  \, \sin^{2d{-}3}\!\theta \, \log (\cos^2\!\theta) \,\mathrm{d}\theta.
\end{equation}
The volume of the $(\ell{-}1)$-dimensional sphere is given by
$S^{(\ell{-}1)}= 2 \pi^{\frac{\ell}{2}}  / \Gamma (\frac{\ell}{2})$.
We are only concerned with even $\ell$, thus
$\Gamma (\frac{\ell}{2}) = \left( \frac{\ell}{2} - 1 \right)!$ and
$S^{(1)}S^{(2d-3)}/S^{(2d-1)} = 2 (d{-}1)$.
We also perform the change of variable $t=\cos^2\theta$,
so Eq.~(\ref{eq:average-entropy-2}) is equal to
\[
  -(\log e)d(d-1)\int_0^1t(1-t)^{d-2}\ln t \,\mathrm{d}t.
\]
Observing that a primitive of $t(1-t)^{d-2}$ is
\[
  \frac{1}{d(d-1)} \; t \; \sum_{k=0}^{d-1}(1-t)^k-\frac{1}{d-1} \;t \,(1-t)^{d-1}
\]
and the derivative of $\ln t$ is $\frac{1}{t}$, and integrating by parts,
we eventually obtain
\[\begin{split}
    & \mathbb{E}_\varphi S(\mathcal{P}(\proj{\varphi})) \\
    &=(\log e)\int_0^1 \left( \sum_{k=0}^{d-1}(1-t)^k-d(1-t)^{d-1} \right) \,\mathrm{d}t \\
    &=(\log e)\sum_{k=2}^{d}\frac{1}{k},
\end{split}\]
concluding the proof.
\end{proof}

\end{document}